\documentclass[a4paper,12pt]{article}
\usepackage[top=2.5cm,bottom=2.5cm,left=2.5cm,right=2.5cm]{geometry}
= 0.0in \headheight = 0.0in \topmargin = 0.3in
\def\Dj{\hbox{D\kern-.73em\raise.30ex\hbox{-}
\raise-.30ex\hbox{}}}
\def\dj{\hbox{d\kern-.33em\raise.80ex\hbox{-}
\raise-.80ex\hbox{\kern-.40em}}}
\usepackage{epsfig}
\usepackage{amsmath,amsthm,amsfonts,amssymb,amscd}

\usepackage{longtable, array}
\usepackage[font={footnotesize}]{caption}
\usepackage[noend]{algpseudocode}
\usepackage[shortcuts]{extdash}
\usepackage[section]{algorithm}
\usepackage[utf8x]{inputenc}
\usepackage{lineno,hyperref}
\usepackage{fancyvrb}
\usepackage{graphicx}
\usepackage{booktabs}
\usepackage{multicol}
\usepackage{balance}
\usepackage{xspace}
\usepackage{float}
\usepackage{ucs}
\usepackage{cite}
\usepackage{arydshln,rotating,lscape,multirow}
\usepackage{mathrsfs,epsfig,amscd,tikz}

\newcommand{\minitab}[2][l]{\begin{tabular}{#1}#2\end{tabular}}
\newtheorem{te}{Theorem}[section]

\newtheorem{co}{Corollary}[section]
\newtheorem{lemma}{Lemma}[section]

\newcommand{\beq}{\begin{eqnarray}}
\newcommand{\eeq}{\end{eqnarray}}
\newcommand{\beqs}{\begin{eqnarray*}}
\newcommand{\eeqs}{\end{eqnarray*}}

\newcommand{\real}{\mathbb R}

\begin{document}
%
%
%
%

\baselineskip=0.30in

\vspace*{20mm}

\begin{center}
{\Large \bf {On extremal trees with respect to the \boldmath $F$-index}\\[2mm]
20 Juli 2015
}
\vspace{10mm}

{\large \bf
Hosam Abdo$^a$\footnote{Corresponding author}, Darko Dimitrov$^b$, Ivan Gutman$^{c,d}$}

\vspace{6mm}

\baselineskip=0.20in

$^a${\it Institut f\"ur Informatik, Freie Universit\"{a}t Berlin,
\\ Takustra{\ss}e 9, D--14195 Berlin, Germany} \\E-mail: {\tt abdo@mi.fu-berlin.de} \\[2mm]
$^b${\it Hochschule f{\" u}r Technik und Wirtschaft Berlin, \\
             Wilhelminenhofstra{\ss}e 75A, D--12459 Berlin, Germany} \\
E-mail: {\tt darko.dimitrov11@gmail.com} \\[2mm]
$^c${\it Faculty of Science, University of Kragujevac, Kragujevac, Serbia} \\
E-mail: {\tt gutman@kg.ac.rs} \\[2mm]
$^d${\it State University of Novi Pazar, Novi Pazar, Serbia}

\vspace{6mm}


\end{center}

\vspace{6mm}

\baselineskip=0.20in

\noindent {\bf Abstract }

\vspace{3mm}

In a study on the structure--dependency of the total $\pi$-electron energy from 1972,
Trinajsti\'c and one of the present authors 
have shown that it depends on the sums $\sum_{v\in V}d(v)^2$ and
$\sum_{v\in V}d(v)^3$, where $d(v)$ is the degree of a vertex $v$ of the \
underling molecular graph $G$. The first sum was later named {\it first Zagreb index}
and over the years became one of the most investigated graph--based molecular
structure descriptors. On the other hand, the second sum, except in very few
works on the general first Zagreb index 
and the zeroth--order general Randi\'c index, 
has been almost completely neglected. Recently, this second sum was named
{\it forgotten index}, or shortly the $F$-{\it index}, and shown to have an
exceptional applicative potential. 
In this paper we examine the trees extremal with respect to the $F$-index.
%
%
%

\baselineskip=0.30in

\vspace{6mm}

\section{Introduction}

The first and the second Zagreb indices, introduced  in 1972 \cite{GT},
are among the oldest graph--based molecular structure descriptors, so-called
topological indices~\cite{BMBM,DB,new1}.
For a graph $G$  with a vertex set $V(G)$ and an edge set $E(G)$,  these are defined as
$$
M_1(G)=\sum_{v\in V(G)}d(v)^2 \quad \mbox{ and }\quad
M_2(G)=\sum_{uv\in E(G)}d(u)d(v),
$$
where the degree of a vertex $v \in V(G)$, denoted by $d(v)$,  is the number of the
first neighbors of $v$.

Over the years these indices have been thoroughly examined and used to study
molecular complexity, chirality, ZE-isomerism and hetero-systems. More about
their physico--chemical applications and mathematical properties can be found
in \cite{NKMT,new1,new3} and \cite{GD,DG,Z,new2,LY}, respectively, as well as
in the references cited therein. The sum of squares of vertex degrees was
also independently studied in quite a few mathematical papers
\cite{Das1,Das2,new4,new5,new6,new7}.

In an early work on the structure--dependency of the total $\pi$-electron
energy \cite{GT}, beside the  first Zagreb index, it was indicated that another
term on which this energy depends is of the form
$$
F(G)=\sum_{v\in V(G)}d(v)^3\,.
$$
For unexplainable reasons, the above  sum, except (implicitly) in a few works
about the general first Zagreb index \cite{ LiZhao-04, LiZheng-05} and
the zeroth--order general Randi\'c index \cite{HULI2005}, has been completely
neglected. Very recently, Furtula and one of the present authors succeeded to
demonstrate that $F(G)$ has a very promising applicative potential \cite{FuGu-15}.
They proposed that $F(G)$ be named the {\it forgotten topological index}, or shortly
the $F$-{\it index}.

Before the publication of the paper \cite{FuGu-15}, the $F$-index was not
examined as such. On the other hand, in some earlier studies on degree--based
graph invariants, it appears as a special case.

The {\it general first Zagreb index} of a graph $G$  is defined as
%
%
%
\beq
M_1^{\alpha}(G)=\sum_{v\in V(G)}d(v)^{\alpha}
= \sum_{uv\in E(G)} \big[d(u)^{\alpha-1} + d(v)^{\alpha-1} \big], \qquad \mbox{ for  }
                                         \alpha \in \real,  \alpha  \neq 0, \alpha \neq 1, \nonumber
\eeq
and its first occurrence in the literature seems to be the work~\cite{LiZhao-04} by Li and Zhao from 2004.
Observe that $M_1^{3}(G)=F(G)$.
In~\cite{LiZhao-04}  the trees with the first three smallest and largest general first Zagreb index were characterized.
In~\cite{ LiZhao-04, LiZheng-05}, among other things, it was shown that
for $\alpha>1$, the star $S_n$ is the tree on $n$ vertices with maximal
$M_1^{\alpha}$ whereas the path $P_n$ is the tree on $n$ vertices with
minimal $M_1^{\alpha}$-value. Needless to say that these results directly
apply to the $F$-index ($\alpha=3$).
More results and information about the  general first Zagreb index can be found
in~\cite{Gutman-14, HULI2005, MB-Liu, LiZhao-04, LiZheng-05, SuSiong-12,ZZ}.

The {\it Randi\'c (or connectivity)} index was introduced by Randi\'c
in $1975$~\cite{r-cmb-75} and is defined as
$$
R = R(G) = \sum_{u v \in E(G)} \frac{1}{\sqrt{d(u)\,d(v)}}\,.
$$
%
%
%
%
Later in $1977$, Kier and Hall~\cite{kh-nsartrmc77} have introduced the
so-called {\it zeroth--order Randi\'c index}
$$
^0R={^0R}(G) =
\sum_{v\in V(G)}d(v)^{-\frac{1}{2}},
$$
which was generalized in $2005$ by Li and Zheng~\cite{LiZheng-05}.
It was named the {\it zeroth--order general Randi\'c index} and was defined as
$$
^0R_{\alpha}={^0R}_{\alpha}(G) =
\sum_{v\in V(G)}d(v)^{\alpha}
$$
for any real number $\alpha$. Evidently, $^0R_{3}(G)=F(G)$.

In~\cite{HULI2005}  graphs of maximal degree at most 4 (molecular graphs ), with given
number of vertices and edges, and with extremal (maximum or minimum) zeroth--order
general Randi\'c index were characterized. This, again, in the special case $\alpha=3$
renders results for the $F$-index.

Here, we extend the work on trees with extremal values of  $F$-index,
by considering trees with bounded maximal degree.
Since the paths are the trees with minimal $F$-index also in this case,
we consider here the characterization of trees with  bounded maximal degree
that have maximal  $F$-value.

In the sequel we introduce notation that will be used in the rest of the paper.
For $u, v, x, y \in V(G)$ such that $uv \in E(G), xy \notin E(G)$, we denote by $G- uv$
the graph that is obtained by deleting the edge $uv$ from $G$ and by $G+xy$
the graph that is obtained by adding the edge $xy$ to $G$.
By $\Delta = \Delta(G)$ we denote the maximal degree of
$G$.
A sequence $D=[d_1, d_2, \ldots, d_n]$ is {\it graphical} if there is a graph
 whose vertex degrees are $d_i$, $i=1,\dots,n$. If in addition
 $d_1 \geq d_2\geq \dots \geq d_n$, then  $D$ is a
 {\it degree sequence}.
 Furthermore, the notation $D(G ) = [x_1^{n_1}, x_2^{n_2},\cdots, x_t^{n_t}]$ means that
the degree sequence is comprised of  $n_i$ vertices of degree $x_i$, where $i = 1, 2, \cdots, t$.

A tree is said to be {\em rooted} if one of its vertices has been designated as the {\em root}.
In a rooted tree, the  {\em parent} of a vertex is the vertex adjacent to it on the path to the root;
every vertex except the root has a unique parent.
A vertex is a parent of a subtree, if this subtree is attached to the vertex.
A  {\em child} of a vertex $v$ is a vertex of which $v$ is the parent.

\medskip

\section[Results]{Results}
First, we characterize trees with maximal degree at most $\Delta$ that have maximal $F$-index.
\subsection{Trees with bounded maximal degree}
%
\begin{te}  \label{FI:MaxGenTree}
Let $T$ be a tree with maximal $F$-index among the trees with $n$ vertices and
maximal degree at most $\Delta$. Then the following holds:
\begin{itemize}
\item[(i)]  If $(n-2)\mod{(\Delta -1})=0$, then $T$ contains $\frac{n-2}{\Delta -1}$
vertices of degree $\Delta$ and  $\frac{n(\Delta-2)+2}{\Delta -1}$ vertices of degree $1$.
\item[(ii)] Otherwise,  $T$ contains $\frac{n-1-x}{\Delta -1}$ vertices of degree
$\Delta$, $\frac{(n-1)(\Delta-2)+x}{\Delta -1}$ vertices of degree $1$ and
one vertex of degree $x$, where $x$ is uniquely determined  by
$2 \leq x \leq \Delta -1$ and $(n-1 -x) \mod{(\Delta -1)=0}$.
\end{itemize}
\end{te}
\begin{proof}
We may assume that $T$ is a rooted tree, whose root is a vertex with degree $\Delta$.
First, we show that  $T$ is comprised of vertices of degrees $\Delta$ and $1$,
and in some cases of one additional vertex $u$, with  $2 \leq d(u) \leq \Delta -1$.
Assume that this is not true and that $T$ has more than one such vertex.
Denote by $V_d$ the set of all vertices of $T$  whose degrees  are different
from $\Delta$ and $1$. We assume that for the degrees of the vertices of $V_d$ the
order $d(u)  \geq \dots  \geq d(v)$ holds. Let $w$ be a child vertex of $v$.

Delete the edge $vw$ and add the edge $uw$ to $T$, obtaining a tree $T^\prime$.
After this transformation, the degree of $u$ increases by one, while the degree of $v$
decreases by one, and the degree set $V_d$ changed into $V_d^\prime$.
It holds that
\beqs
F(T^{\prime}) - F(T) =  (d(u)+1)^3-d(u)^3+(d(w)-1)^3-d(w)^3 > 0\,.
\eeqs
It holds that $|V_d^\prime| \leq |V_d|$, with strict inequality if $d(u)=\Delta -1$ or $d(v)=2$.
If $|V_d^\prime| \geq 2$ we choose from $V_d^\prime$ a vertex with maximal and a
vertex with minimal degree, and we repeat the
above operation, obtaining a tree with larger $F$-index.
We proceed on iteratively with the same type of transformation until the transformed $V_d$
has cardinality $1$ or $0$, and thus obtain a contradiction to the assumption
that  $T$ has more that one vertex with degree different from
$1$ and $\Delta$.
It follows that $T$ has a degree sequence $[\Delta^{n_{ \Delta}}, 1^{{n_1}}]$ or
$[\Delta^{n_{ \Delta}}, x^{1}, 1^{{n_1}}]$,
$2 \leq x \leq \Delta -1$.

Observe that any transformation on $T$ will result in a tree with $|V_d|>1$ or will disconnect $T$.
Thus, we conclude that $T$ must have one of the two above presented degree sequences.

Next, with respect to the cardinality of $|V_d|$, we determine the parameters
$n_\Delta$\,, $n_1$\,, and $x$.

\smallskip
\noindent
{\bf Case~$1$.}  $|V_d|=0.$
In this case
the degree sequence of $T$ is $[\Delta^{n_{ \Delta}}, 1^{{n_1}}]$ where the equations
\beq \label{deggen0}
&&\sum_{i=1}^ {\Delta} \, n_i = n_{ \Delta} + n_1  = n \qquad  \mbox{and} \qquad
 \sum_{i=1}^ {\Delta} \, i \; n_i =  \Delta n_{ \Delta} + n_1  = 2(n-1) \nonumber
\eeq
hold.
The above equations have the integer solution
\beqs
n_\Delta =  \frac{n-2}{\Delta -1}, \quad\quad n_1 = \frac{n(\Delta-2)+2}{\Delta -1},
\eeqs
for $(n-2)\mod{(\Delta -1})=0$.

\smallskip
\noindent
{\bf Case~$2.$}  $|V_d|=1.$
Here the degree sequence of $T$ is $[\Delta^{n_{ \Delta}}, x^{1}, 1^{{n_1}}]$, with
\beq \label{deggen}
\sum_{i=1}^ {\Delta} \, n_i = n_{ \Delta} + n_1 + 1 = n \qquad  \mbox{and} \qquad
  \sum_{i=1}^ {\Delta} \, i \; n_i =  \Delta n_{ \Delta} + n_1 + x  = 2(n-1). \nonumber
\eeq
The above equations give the integer solution
\beqs
n_\Delta =  \frac{n-1-x}{\Delta -1}, \quad\quad n_1 = \frac{(n-1)(\Delta-2)+x}{\Delta -1},
\eeqs
for $(n-1 -x) \mod{(\Delta -1)=0}$.
\end{proof}

As straightforward consequence of Theorem~\ref{FI:MaxGenTree}, we obtain the
maximal value of the $F$-index for trees with maximal degree $\Delta$.

\begin{co}  \label{co-FI:MaxGenTree1}
Let $T$ be a tree with maximal $F$-index among the trees with $n$ vertices and
maximal degree at most $\Delta$. Then,
\begin{itemize}
\item[(i)]  if $(n-2)\mod{(\Delta -1})=0$,
\beqs
 F(T) =   \Delta( \Delta+1)(n-2)+2(n-1);
 \eeqs
\item[(ii)] otherwise,
 \beqs
 F(T) =  ( \Delta^2 + \Delta +2)(n-1) - ( \Delta^2 + \Delta +1) x + x^3,
 \eeqs
 where $x$ is uniquely determined by
 $2 \leq x \leq \Delta -1$ and $n-1 -x \mod{(\Delta -1)=0}$.
\end{itemize}
\end{co}
\begin{proof}
If $(n-2)\mod{(\Delta -1})=0$, by Theorem~\ref{FI:MaxGenTree}, $T$ has degree sequence
$[\Delta^{n_{ \Delta}}, 1^{{n_1}}]$, where $n_{ \Delta}=(n-2)/(\Delta -1)$ and
 $n_1=(n(\Delta-2)+2)/(\Delta -1)$. Thus,
\beqs
 F(T) &=& \Delta^3 \; \frac{n-2}{\Delta -1} +  \frac{n(\Delta-2)+2}{\Delta -1} = \Delta( \Delta+1)(n-2)+2(n-1).
\eeqs

\noindent
Otherwise, $T$ has degree sequence $[\Delta^{n_{ \Delta}}, x^{1}, 1^{{n_1}}]$, where
$n_{ \Delta}=(n-1-x)/(\Delta -1)$, $n_1=((n-1)(\Delta-2)+x)/(\Delta -1)$.
Then,
 \beqs
 F(T) &=& \Delta^3 \; \frac{n-(x+1)}{\Delta -1} +  \frac{(n-1)(\Delta-2)+x}{\Delta -1} + x^3  \\
                     &=& \frac{(\Delta^3 +\Delta -2) n +  (1-\Delta^3 ) x - (\Delta^3 +\Delta -2)}{\Delta -1} + x^3  \\
                     &=& ( \Delta^2 + \Delta +2)(n-1) - ( \Delta^2 + \Delta +1) x + x^3.
\eeqs
\end{proof}

%

\subsection{Molecular trees}

\noindent
For the special case of molecular trees, i.e., $\Delta \leq 4$, we obtain the following result.
Recall that such trees provide the graph representation of the so-called saturated
hydrocarbons or alkanes \cite{GuPo,BMBM} and are of major importance in theoretical
chemistry.

\begin{te} \label{F-index-trees}
Let $T$ be a molecular tree with maximal $F$-index among the trees with $n$ vertices.
Then the following holds:
\begin{itemize}
\item  If $(n-2)\mod{3}=0$, then $T$ contains $\frac{n-2}{3}$ vertices of degree
$4$ and $\frac{2n+2}{3}$ vertices of degree $1$.  Its $F$-index is
$$F(T)= 22n-42.$$
\item Otherwise,  $T$ contains $\frac{n-1-x}{3}$ vertices of degree
$4$, $\frac{2(n-1)+x}{3}$ vertices of degree $1$ and
 one vertex of degree $x$, where $x$ is uniquely determined  by
 $2 \leq x \leq 3$ and $(n-1 -x) \mod{3=0}$. Its $F$-index is
$$F(T)= 22(n-1)-21x+x^3.$$
\end{itemize}
\end{te}
\begin{proof}
For  $n=2,3,4,5$ the star $S_n$ maximizes the $F$-index
\cite{LiZhao-04} .
 For $n=6$ and $\Delta \leq 4$, the possible degree sequences of $T$ are
 $(4,2,1,1,1,1), (3,3,1,1,1,1), (3,2,2,1,1,1)$ and $(2,2,2,2,1,1)$.
The first of these degree sequences corresponds to the largest $F$-index, $F(T)$.

Thus, the corresponding  degree sequences of $T$ for $n=2,3,4,5,6$ satisfy the
theorem.

In the rest of the proof we assume that  $n \geq 7$.
First, we show that $T$ has maximal degree $\Delta=4$.
A tree with maximal degree $\Delta$ we denote by $T_{\Delta}$.
By Corollary~\ref{co-FI:MaxGenTree1}$(i)$, $F(T_{2})=8n-14$.
Further, we distinguish between two case.

\smallskip
\noindent
{\bf Case 1.}  $(n-2) \mod 2 =0$.

\noindent
By Corollary~\ref{co-FI:MaxGenTree1}$(i)$, we have that $F(T_{3})=14n -26$.
Since $(n-2) \mod 2 =0$, it follows that $(n-2) \mod 3 \neq 0$, and by 
Corollary~\ref{co-FI:MaxGenTree1}$(ii)$,
$F(T_{4})=22(n-1)-21x+x^3$, where $x=2$ or $x=3$.
For $n \geq 7$,  $F(T_{4}) >F(T_{3}) > F(T_{2})$ is satisfied.

\smallskip
\noindent
{\bf Case 2.} $(n-2) \mod 2 \neq 0$.

\noindent
By  Corollary~\ref{co-FI:MaxGenTree1}$(ii)$, $F(T_{3})=14(n-1)-13x+x^3$, where $x=2$ or $x=3$.

\smallskip
\noindent
{\bf Subcase 2.1.} $(n-2) \mod 3 = 0$.

\noindent
By Corollary~\ref{co-FI:MaxGenTree1}$(i)$,
$F(T_{4})=22n-42$.
For $n \geq 7$ and $x=2,3$, it holds that $F(T_{4}) >F(T_{3}) > F(T_{2})$.

\smallskip
\noindent
{\bf Subcase 2.2.} $(n-2) \mod 3 \neq 0$.

\noindent
By Corollary~\ref{co-FI:MaxGenTree1}$(ii)$,
$F(T_{4})=22(n-1)-21x+x^3$.
Again, a straightforward calculation yields that $F(T_{4}) >F(T_{3}) > F(T_{2})$.

\smallskip
\noindent
So, we have shown  that for $n \geq 3,4,5$  the tree with maximal $F$-index
has maximal degree  $n-1$, and for $n \geq 5$ it has  maximal degree $\Delta=4$.
Setting $\Delta=4$ in Theorem~\ref{FI:MaxGenTree} and Corollary~\ref{co-FI:MaxGenTree1},
we complete the proof.
\end{proof}

\noindent
We would like to note that the results of Theorem~\ref{F-index-trees} coincide
with the results of Theorem~$2.2$ in \cite{HULI2005}, pertaining to the zeroth--order 
general Randi\'c index of a graph, when $\alpha = 3$ and $m=n-1$.

In the sequel we present some computational results obtain by integer linear programming 
and by exhaustive computer search of trees with maximal $F$-index.

\subsection[ILP Model]{Computational results} \label{Sub:LPmodel}
%

Degree sequences of trees with bounded maximal degree and maximal $F(T)$ can
be completely described by solving the integer linear programming problem 
in Lemma \ref{FILP:Lem}.
We denote the number of vertices of tree $T$ of degree $i$ by $n_i$, $\forall \, i = 1, 2, \cdots, \Delta$
and the number of edges, connecting a vertex of degree $i$ with a vertex of degree $j$, by $m_{ij}$.
\begin{lemma}  \label{FILP:Lem}
Let $T$ be a tree with maximal $F$-index among all trees with $n$ vertices and maximal degree $\Delta$.
Then, the trees with maximal $F$-index are completely described by solving the maximization problem
%
\begin{alignat*}{2}
    \text{maximize } ~&F = \sum_{i = 1}^{\Delta} \;\;  i^3  \; n_i  \\
    \text{subject to } ~& \sum_{i = 1}^{\Delta }  \, n_i \; =\; n ,         \\
                                 & \sum_{i = 1}^{\Delta }  \, i \, n_i \;=\;  2 (n-1),    \\
                                 & \sum_{i= 2 } ^{\Delta-1 }  \, n_i  \; \leq \; 1,          \\
                                 & \sum_{j = 1, \,  j \neq i } ^{\Delta } m_{ij}  + 2 m_{ii}  \; =\; i n_i   ~~~~~~~  1 \; \leq\; i \; \leq \; \Delta,    \\
                                 & 0 \; \leq\;  n_i \; \leq \;  n-1 ,  ~~~~~~~~~~~~~~~~1 \; \leq\; i \; \leq \; \Delta,  \\
                                 & 0 \; \leq\;  m_{ij} \; \leq \;  n - 2 ,  ~~~~~~~~~~~~~~~ 1 \; \leq\; i, \, j \; \leq \; \Delta.
\end{alignat*}

\end{lemma}
 %

For experimental proposes we first consider the molecular trees of orders $n = 4, \cdots, 20$.
In this case, the previous integer linear programming model can be represented as follows.
\beq \label{chem-maxFI}
\left.
\begin{aligned}
&\mbox{maximize} &
                           F = \; &n_1&  &+& 8  &n_2&  &+& 27 &n_3&   &+& 64 &n_4&   & &                       \\ 
&\qquad \mbox{s. t. }&   &n_1&  &+&     &n_2&  &+&      &n_3&   &+&      & n_4& &             &= &\; n,\\
&                            &   &n_1&  &+& 2 &n_2&   &+& 3   &n_3&   &+& 4   & n_4&  &               &= &\;2 (n-1),\\
&                            &   &      &  &  &    &n_2&   &+&      &n_3&   &  &      &       &   & &   \leqslant&\;1,\\ 
&                            & - &n_1&  &+& &m_{12}& &+& &m_{13}& &+&     & m_{14}& &             &= &\;0,\\
&                            & -2 &n_2&  &+& &m_{12}& &+& 2&m_{22}& &+& & m_{23}&+& m_{24}&= &\;0,\\
&                            & -3 &n_3&  &+& &m_{13}& &+& &m_{23}& &+& & 2m_{33}&+& m_{34}&= &\;0, \\
&                            & -4 &n_4&  &+& &m_{14}& &+& &m_{24}& &+& & m_{34}&+& 2 m_{44}&=&\;0,   
\end{aligned}
\right\} \\
   0  \leqslant  n_i  \leqslant  n-1,  \quad  0  \leqslant  m_{ij}   \leqslant n-2,  \quad  1  \leq  i  \leq  4, \quad  1  \leq  j \leq  4. \qquad\qquad  \nonumber
\eeq
The above integer linear program was implemented in Matlab\cite{matlab-soft},
and  run on $2.3$ GHz Intel Core i$5$ processor with $4$GB 1333 MHz DDR3 RAM.
Beside the degree sequence that maximize the $F$-index,
it gives only one corresponding tree. However, for a given degree sequence there usually 
exist several non-isomorphic trees with the same degree sequence.
In view of this, the degree sequence as well as all corresponding trees were obtained by 
exhaustive computer search using the mathematical software
Sage \cite{sage}. In the contrast to the integer linear program that gives a solution 
for $n$ of order $1000000$ and $\Delta \leq 400$ in less than a minute,
the exhaustive computer search with Sage for $n=20$  took several hours.
On the other hand, due to the closed--form solutions in Theorem~\ref{FI:MaxGenTree},
we obtain the degree sequence that maximizes the $F$-index for
arbitrary $n$ and arbitrary $\Delta$ in few milliseconds.

The experimental results for $\Delta=4$ and $\Delta=5$  are given in 
Tables \ref{tableChemD4} and \ref{tableTreesD5}, respectively,
while their corresponding trees are given in Figures \ref{pr_sl} and \ref{TreesDelta5}.

%
%
\newpage
\begin{center}
\setlength{\extrarowheight}{0.0cm}
\begin{longtable}{|c|p{2.6cm}p{0.6cm}|p{0.6cm}p{0.6cm}p{0.6cm}p{0.6cm}p{0.9cm}p{0.9cm}p{0.9cm}p{0.3cm}|c|}
\hline 
\multicolumn{1}{|c|}{\multirow{2}{*}{$n$}}  &  \multicolumn{2}{c|}{\textbf{Sage}} &  \multicolumn{8}{c|}{\textbf{LP solution - Matlab}} & \multicolumn{1}{c|}{\multirow{2}{*}{\textbf{$F$}}} \\
\cline{2-11}\rule[2mm]{0mm}{1.5mm}
  & \multicolumn{1}{c}{\textbf{$D(T)$}} & \multicolumn{1}{c|}{\textbf{$\#T$}} & \multicolumn{8}{c|}{\textbf{Non-zero variables}} &  \\ \hline
\endhead
\multirow{2}{*}{4} & \multirow{2}{*}{$[3^{1}, 1^{3}]$}&\multirow{2}{*}{1}&\multirow{2}*{\minitab[c]{$n_1$ \\ 3}} & \multirow{2}*{\minitab[c]{$n_3$\\ 1}} & &&
                                \multirow{2}*{\minitab[c]{$m_{1,3}$ \\3}} &&&&\multirow{2}{*}{30}  \rule[-2mm]{0mm}{1mm}\\
&& &&&&&&&&&  \rule[-2mm]{0mm}{1mm}\\
\multirow{2}{*}{5} & \multirow{2}{*}{$[4^{1}, 1^{4}]$}&\multirow{2}{*}{1}&\multirow{2}*{\minitab[c]{$n_1$ \\ 4}} & \multirow{2}*{\minitab[c]{$n_4$\\ 1}} & &&
                                \multirow{2}*{\minitab[c]{$m_{1,4}$ \\4}} &&&&\multirow{2}{*}{68}  \rule[-2mm]{0mm}{1mm}\\
&& &&&&&&&&&  \rule[-2mm]{0mm}{1mm}\\
\multirow{2}{*}{6} & \multirow{2}{*}{$[4^{1}, 2^{1}, 1^{4}]$}&\multirow{2}{*}{1}&\multirow{2}*{\minitab[c]{$n_1$ \\ 4}} & \multirow{2}*{\minitab[c]{$n_2$\\ 1}} & \multirow{2}*{\minitab[c]{$n_4$\\ 1}} &&                             \multirow{2}*{\minitab[c]{$m_{1,2}$ \\1}} &\multirow{2}*{\minitab[c]{$m_{1,4}$ \\ 3}} &\multirow{2}*{\minitab[c]{$m_{2,4}$ \\ 1}}   &&\multirow{2}{*}{76}  \rule[-2mm]{0mm}{1mm}\\
&& &&&&&&&&&  \rule[-2mm]{0mm}{1mm}\\
\multirow{2}{*}{7} & \multirow{2}{*}{$[4^{1}, 3^{1}, 1^{5}]$}&\multirow{2}{*}{1}&\multirow{2}*{\minitab[c]{$n_1$ \\ 5}} & \multirow{2}*{\minitab[c]{$n_3$\\ 1}} & \multirow{2}*{\minitab[c]{$n_4$\\ 1}} &&                             \multirow{2}*{\minitab[c]{$m_{1,3}$ \\2}} &\multirow{2}*{\minitab[c]{$m_{1,4}$ \\ 3}} &\multirow{2}*{\minitab[c]{$m_{3,4}$ \\ 1}}   &&\multirow{2}{*}{96}  \rule[-2mm]{0mm}{1mm}\\
&& &&&&&&&&&  \rule[-2mm]{0mm}{1mm}\\
\multirow{2}{*}{8} & \multirow{2}{*}{$[4^{2}, 1^{6}]$}&\multirow{2}{*}{1}&\multirow{2}*{\minitab[c]{$n_1$ \\ 6}} & \multirow{2}*{\minitab[c]{$n_4$\\ 2}} & &&
                                \multirow{2}*{\minitab[c]{$m_{1,4}$ \\6}} &   \multirow{2}*{\minitab[c]{$m_{4,4}$ \\1}} &&& \multirow{2}{*}{134}  \rule[-2mm]{0mm}{1mm}\\
&& &&&&&&&&&  \rule[-2mm]{0mm}{1mm}\\
\multirow{2}{*}{9} & \multirow{2}{*}{$[4^{2}, 2^{1}, 1^{6}]$}&\multirow{2}{*}{2}&\multirow{2}*{\minitab[c]{$n_1$ \\ 6}} & \multirow{2}*{\minitab[c]{$n_2$\\ 1}} & \multirow{2}*{\minitab[c]{$n_4$\\ 2}} &&                             \multirow{2}*{\minitab[c]{$m_{1,4}$ \\6}} &\multirow{2}*{\minitab[c]{$m_{2,4}$ \\ 2}} &&& \multirow{2}{*}{142}  \rule[-2mm]{0mm}{1mm}\\
&& &&&&&&&&&  \rule[-2mm]{0mm}{1mm}\\
\multirow{2}{*}{10} &\multirow{2}{*}{$[4^{2}, 3^{1}, 1^{7}]$}  &\multirow{2}{*}{2} & \multirow{2}*{\minitab[c]{$n_1$ \\ 7}}  & \multirow{2}*{\minitab[c]{$n_3$\\ 1}}  & \multirow{2}*{\minitab[c]{$n_4$\\ 2}}
                             & & \multirow{2}*{\minitab[c]{$m_{1,3}$ \\1}} &\multirow{2}*{\minitab[c]{$m_{1,4}$ \\ 6}} & \multirow{2}*{\minitab[c]{$m_{3,4}$ \\ 2}}& &\multirow{2}{*}{162}  \rule[-2mm]{0mm}{1mm}\\
&& &&&&&&&&&   \rule[-2mm]{0mm}{1mm}\\
%
\multirow{2}{*}{11} &\multirow{2}{*}{$[4^{3}, 1^{8}]$} &\multirow{2}{*}{1} & \multirow{2}*{\minitab[c]{$n_1$ \\ 8}}  & \multirow{2}*{\minitab[c]{$n_4$\\ 3}}  & & &
                                 \multirow{2}*{\minitab[c]{$m_{1,4}$ \\8}} &\multirow{2}*{\minitab[c]{$m_{4,4}$ \\ 2}} && &\multirow{2}{*}{200}  \rule[-2mm]{0mm}{1mm}\\
&& &&&&&&&&&   \rule[-2mm]{0mm}{1mm}\\
%
\multirow{2}{*}{12} & \multirow{2}{*}{$[4^{3}, 2^{1}, 1^{8}]$}&\multirow{2}{*}{3}&\multirow{2}*{\minitab[c]{$n_1$ \\ 8}} & \multirow{2}*{\minitab[c]{$n_2$\\ 1}} & \multirow{2}*{\minitab[c]{$n_4$\\ 3}} &                &\multirow{2}*{\minitab[c]{$m_{1,4}$ \\ 8}} &\multirow{2}*{\minitab[c]{$m_{2,4}$ \\ 2}}   &\multirow{2}*{\minitab[c]{$m_{4,4}$ \\ 1}}&&\multirow{2}{*}{208}  \rule[-2mm]{0mm}{1mm}\\
&& &&&&&&&&&  \rule[-2mm]{0mm}{1mm}\\
%
\multirow{2}{*}{13} & \multirow{2}{*}{$[4^{3}, 3^{1}, 1^{9}]$}&\multirow{2}{*}{4}&\multirow{2}*{\minitab[c]{$n_1$ \\ 9}} & \multirow{2}*{\minitab[c]{$n_3$\\ 1}} & \multirow{2}*{\minitab[c]{$n_4$\\ 3}} &&                             \multirow{2}*{\minitab[c]{$m_{1,4}$ \\9}} &\multirow{2}*{\minitab[c]{$m_{3,4}$ \\ 3}} &  &&\multirow{2}{*}{228}  \rule[-2mm]{0mm}{1mm}\\
&& &&&&&&&&&  \rule[-2mm]{0mm}{1mm}\\
%
\multirow{2}{*}{14} &\multirow{2}{*}{$[4^{4}, 1^{10}]$}&\multirow{2}{*}{2} & \multirow{2}*{\minitab[c]{$n_1$ \\ 10}}  & \multirow{2}*{\minitab[c]{$n_4$\\ 4}}  &  &  &
                                 \multirow{2}*{\minitab[c]{$m_{1,4}$ \\10}} &\multirow{2}*{\minitab[c]{$m_{4,4}$ \\ 3}} && &\multirow{2}{*}{266}  \rule[-2mm]{0mm}{1mm}\\
&& &&&&&&&&&   \rule[-2mm]{0mm}{1mm}\\
%
\multirow{2}{*}{15} &\multirow{2}{*}{$[4^{4}, 2^{1}, 1^{10}]$}&\multirow{2}{*}{6} & \multirow{2}*{\minitab[c]{$n_1$ \\ 10}}  & \multirow{2}*{\minitab[c]{$n_2$\\ 1}}  &\multirow{2}*{\minitab[c]{$n_4$\\ 4}} &  &\multirow{2}*{\minitab[c]{$m_{1,4}$ \\ 10}} &\multirow{2}*{\minitab[c]{$m_{2,4}$ \\ 2}}&\multirow{2}*{\minitab[c]{$m_{4,4}$ \\ 2}}  & &\multirow{2}{*}{274}  \rule[-2mm]{0mm}{1mm}\\
&& &&&&&&&&&   \rule[-2mm]{0mm}{1mm}\\
\multirow{2}{*}{16} & \multirow{2}{*}{$[4^{4}, 3^{1}, 1^{11}]$} &\multirow{2}{*}{8}&\multirow{2}*{\minitab[c]{$n_1$ \\ 11}} & \multirow{2}*{\minitab[c]{$n_3$\\ 1}} & \multirow{2}*{\minitab[c]{$n_4$\\ 4}} &&                             \multirow{2}*{\minitab[c]{$m_{1,4}$ \\11}} &\multirow{2}*{\minitab[c]{$m_{3,4}$ \\3}} & \multirow{2}*{\minitab[c]{$m_{4,4}$ \\ 1}} &&\multirow{2}{*}{294}  \rule[-2mm]{0mm}{1mm}\\
&& &&&&&&&&&  \rule[-2mm]{0mm}{1mm}\\
\multirow{2}{*}{17} & \multirow{2}{*}{$[4^{5}, 1^{12}]$}&\multirow{2}{*}{3}&\multirow{2}*{\minitab[c]{$n_1$ \\ 12}} & \multirow{2}*{\minitab[c]{$n_4$\\ 5}} & &&                             \multirow{2}*{\minitab[c]{$m_{1,4}$ \\12}} &\multirow{2}*{\minitab[c]{$m_{4,4}$ \\ 4}} &  &&\multirow{2}{*}{332}  \rule[-2mm]{0mm}{1mm}\\
&& &&&&&&&&&  \rule[-2mm]{0mm}{1mm}\\
\multirow{2}{*}{18} &\multirow{2}{*}{$[4^{5}, 2^{1}, 1^{12}]$} &\multirow{2}{*}{14} & \multirow{2}*{\minitab[c]{$n_1$ \\ 12}}  & \multirow{2}*{\minitab[c]{$n_2$\\ 1}}  &\multirow{2}*{\minitab[c]{$n_4$\\ 5}}  &
                                &\multirow{2}*{\minitab[c]{$m_{1,4}$ \\ 12}} &\multirow{2}*{\minitab[c]{$m_{2,4}$ \\ 2}}&\multirow{2}*{\minitab[c]{$m_{4,4}$ \\ 3}}& &\multirow{2}{*}{340}  \rule[-2mm]{0mm}{1mm}\\
&& &&&&&&&&&   \rule[-2mm]{0mm}{1mm}\\
\multirow{2}{*}{19} &\multirow{2}{*}{$[4^{5}, 3^{1}, 1^{13}]$}&\multirow{2}{*}{17} & \multirow{2}*{\minitab[c]{$n_1$ \\ 13}}  & \multirow{2}*{\minitab[c]{$n_3$\\ 1}}  &\multirow{2}*{\minitab[c]{$n_4$\\ 5}}  &  &
                                 \multirow{2}*{\minitab[c]{$m_{1,4}$ \\13}} &\multirow{2}*{\minitab[c]{$m_{3,4}$ \\ 3}} &\multirow{2}*{\minitab[c]{$m_{4,4}$ \\ 2}}& &\multirow{2}{*}{360}  \rule[-2mm]{0mm}{1mm}\\
&&&&&&&&&&&   \rule[-2mm]{0mm}{1mm}\\
\multirow{2}{*}{20} & \multirow{2}{*}{$[4^{6}, 1^{14}]$}&\multirow{2}{*}{5}&\multirow{2}*{\minitab[c]{$n_1$ \\ 14}} & \multirow{2}*{\minitab[c]{$n_4$\\ 6}} &  &&                             \multirow{2}*{\minitab[c]{$m_{1,4}$ \\14}} &\multirow{2}*{\minitab[c]{$m_{4,4}$ \\ 5}} &  &&\multirow{2}{*}{398}  \rule[-2mm]{0mm}{1mm}\\
&& &&&&&&&&&  \rule[-2mm]{0mm}{1mm}\\
  \hline
\caption[Extremal molecular trees]{Extremal molecular trees of order up to  $20$.
First column: the order of a tree.
Second column: the degree sequences that maximize the $F$-index and the number of corresponding trees, obtained by Sage.
Third column:  solutions of an integer linear programming problem with one realization of each degree sequence, obtained by Matlab.
Fourt column: the value of the $F$-index.
}
\label{tableChemD4}
\end{longtable}
\end{center}
%
\begin{figure}[h!pt]
\begin{center}
\vskip-1cm
\includegraphics[scale=1.0]{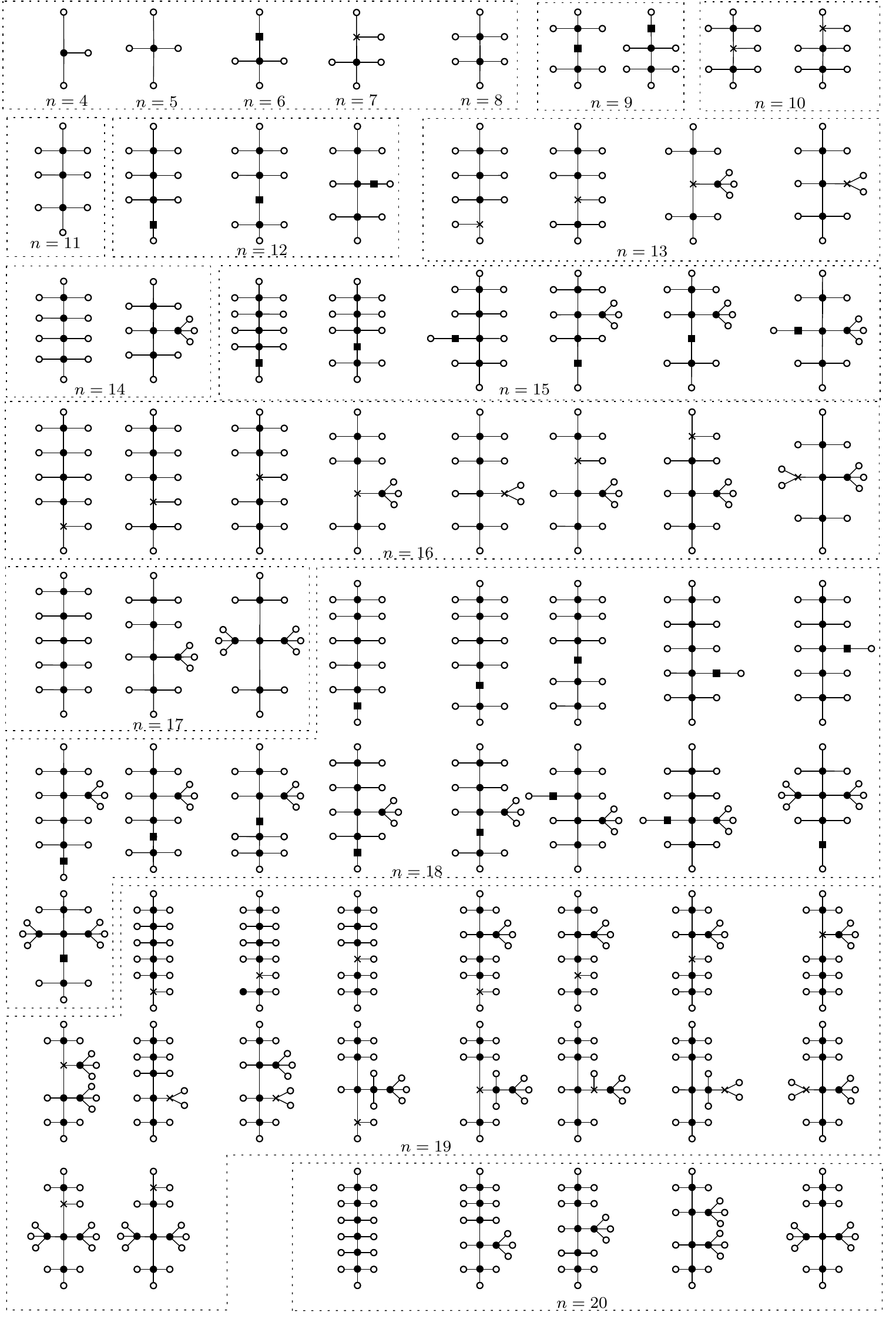}
\caption{Extremal molecular trees of order $n = 4, \cdots, 20$ with respect to the $F$-index.}
\label{pr_sl}
\vskip-0.5cm
\end{center}
\end{figure}
%
\newpage
\begin{center}
\setlength{\extrarowheight}{0.0cm}
\begin{longtable}{|c|p{2.5cm}p{0.5cm}|p{0.6cm}p{0.6cm}p{0.6cm}p{0.6cm}p{0.9cm}p{0.9cm}p{0.9cm}p{0.9cm}|c|}
\hline
\multicolumn{1}{|c|}{\multirow{2}{*}{$n$}}  &  \multicolumn{2}{c|}{\textbf{Sage}} &  \multicolumn{8}{c|}{\textbf{LP solution - Matlab}} & \multicolumn{1}{c|}{\multirow{2}{*}{\textbf{$F$}}} \\
\cline{2-11}\rule[2mm]{0mm}{1.5mm}
  & \multicolumn{1}{c}{\textbf{$D(T)$}} & \multicolumn{1}{c|}{\textbf{$\#T$}} & \multicolumn{8}{c|}{\textbf{Non-zero variables}} &  \\ \hline
\endhead
\multirow{2}{*}{4} & \multirow{2}{*}{$[3^{1}, 1^{3}]$}&\multirow{2}{*}{1}&\multirow{2}*{\minitab[c]{$n_1$ \\ 3}} & \multirow{2}*{\minitab[c]{$n_3$\\ 1}} & &&
                                \multirow{2}*{\minitab[c]{$m_{1,3}$ \\3}} &&&&\multirow{2}{*}{30}  \rule[-2mm]{0mm}{1mm}\\
&& &&&&&&&&&  \rule[-2mm]{0mm}{1mm}\\
\multirow{2}{*}{5} & \multirow{2}{*}{$[4^{1}, 1^{4}]$}&\multirow{2}{*}{1}&\multirow{2}*{\minitab[c]{$n_1$ \\ 4}} & \multirow{2}*{\minitab[c]{$n_4$\\ 1}} & &&
                                \multirow{2}*{\minitab[c]{$m_{1,4}$ \\4}} &&&&\multirow{2}{*}{68}  \rule[-2mm]{0mm}{1mm}\\
&& &&&&&&&&&  \rule[-2mm]{0mm}{1mm}\\
\multirow{2}{*}{6} & \multirow{2}{*}{$[5^{1}, 1^{5}]$}&\multirow{2}{*}{1}&\multirow{2}*{\minitab[c]{$n_1$ \\ 5}} & \multirow{2}*{\minitab[c]{$n_5$\\ 1}} & &&                             \multirow{2}*{\minitab[c]{$m_{1,5}$ \\5}} && &&\multirow{2}{*}{130}  \rule[-2mm]{0mm}{1mm}\\
&& &&&&&&&&&  \rule[-2mm]{0mm}{1mm}\\
\multirow{2}{*}{7} & \multirow{2}{*}{$[5^{1}, 2^{1}, 1^{5}]$}&\multirow{2}{*}{1}&\multirow{2}*{\minitab[c]{$n_1$ \\ 5}} & \multirow{2}*{\minitab[c]{$n_2$\\ 1}} & \multirow{2}*{\minitab[c]{$n_5$\\ 1}} &&                             \multirow{2}*{\minitab[c]{$m_{1,2}$ \\1}} &\multirow{2}*{\minitab[c]{$m_{1,5}$ \\ 4}} &\multirow{2}*{\minitab[c]{$m_{2,5}$ \\ 1}}   &&\multirow{2}{*}{138}  \rule[-2mm]{0mm}{1mm}\\
&& &&&&&&&&&  \rule[-2mm]{0mm}{1mm}\\
\multirow{2}{*}{8} & \multirow{2}{*}{$[5^{1}, 3^{1}, 1^{6}]$}&\multirow{2}{*}{1}&\multirow{2}*{\minitab[c]{$n_1$ \\ 6}} & \multirow{2}*{\minitab[c]{$n_3$\\ 1}} &  \multirow{2}*{\minitab[c]{$n_5$\\ 1}}&&
                                \multirow{2}*{\minitab[c]{$m_{1,3}$ \\2}} &   \multirow{2}*{\minitab[c]{$m_{1,5}$ \\4}} &\multirow{2}*{\minitab[c]{$m_{3,5}$ \\1}}&& \multirow{2}{*}{158}  \rule[-2mm]{0mm}{1mm}\\
&& &&&&&&&&&  \rule[-2mm]{0mm}{1mm}\\
\multirow{2}{*}{9} & \multirow{2}{*}{$[5^{1}, 4^{1}, 1^{7}]$}&\multirow{2}{*}{1}&\multirow{2}*{\minitab[c]{$n_1$ \\ 7}} & \multirow{2}*{\minitab[c]{$n_4$\\ 1}} & \multirow{2}*{\minitab[c]{$n_5$\\ 1}} && \multirow{2}*{\minitab[c]{$m_{1,4}$ \\3}} &\multirow{2}*{\minitab[c]{$m_{1,5}$ \\ 4}} &\multirow{2}*{\minitab[c]{$m_{4,5}$ \\ 1}} && \multirow{2}{*}{196}  \rule[-2mm]{0mm}{1mm}\\
&& &&&&&&&&&  \rule[-2mm]{0mm}{1mm}\\
\multirow{2}{*}{10} &\multirow{2}{*}{$[5^{2}, 1^{8}]$}  &\multirow{2}{*}{1} & \multirow{2}*{\minitab[c]{$n_1$ \\ 8}}  & \multirow{2}*{\minitab[c]{$n_5$\\ 2}}  &
                             & & \multirow{2}*{\minitab[c]{$m_{1,5}$ \\8}} &\multirow{2}*{\minitab[c]{$m_{5,5}$ \\ 1}} & & &\multirow{2}{*}{258}  \rule[-2mm]{0mm}{1mm}\\
&& &&&&&&&&&   \rule[-2mm]{0mm}{1mm}\\
%
\multirow{2}{*}{11} &\multirow{2}{*}{$[5^{2}, 2^{1}, 1^{8}]$} &\multirow{2}{*}{1} & \multirow{2}*{\minitab[c]{$n_1$ \\ 8}}  & \multirow{2}*{\minitab[c]{$n_2$\\ 1}}  & \multirow{2}*{\minitab[c]{$n_5$ \\ 2}}  & &
                                 \multirow{2}*{\minitab[c]{$m_{1,5}$ \\8}} &\multirow{2}*{\minitab[c]{$m_{2,5}$ \\ 2}} && &\multirow{2}{*}{266}  \rule[-2mm]{0mm}{1mm}\\
&& &&&&&&&&&   \rule[-2mm]{0mm}{1mm}\\
%
\multirow{2}{*}{12} & \multirow{2}{*}{$[5^{2}, 3^{1}, 1^{9}]$}&\multirow{2}{*}{2}&\multirow{2}*{\minitab[c]{$n_1$ \\ 9}} & \multirow{2}*{\minitab[c]{$n_3$\\ 1}} & \multirow{2}*{\minitab[c]{$n_5$\\ 2}} &                &\multirow{2}*{\minitab[c]{$m_{1,3}$ \\ 2}} &\multirow{2}*{\minitab[c]{$m_{1,5}$ \\ 7}}   &\multirow{2}*{\minitab[c]{$m_{3,5}$ \\ 1}}&\multirow{2}*{\minitab[c]{$m_{5,5}$ \\ 1}}&\multirow{2}{*}{286}  \rule[-2mm]{0mm}{1mm}\\
&& &&&&&&&&&  \rule[-2mm]{0mm}{1mm}\\
%
\multirow{2}{*}{13} & \multirow{2}{*}{$[5^{2}, 4^{1}, 1^{10}]$}&\multirow{2}{*}{2}&\multirow{2}*{\minitab[c]{$n_1$ \\ 10}} & \multirow{2}*{\minitab[c]{$n_4$\\ 1}} & \multirow{2}*{\minitab[c]{$n_5$\\ 2}} &&                             \multirow{2}*{\minitab[c]{$m_{1,4}$ \\2}} &\multirow{2}*{\minitab[c]{$m_{1,5}$ \\ 8}} &  \multirow{2}*{\minitab[c]{$m_{4,5}$ \\2}}  &&\multirow{2}{*}{324}  \rule[-2mm]{0mm}{1mm}\\
&& &&&&&&&&&  \rule[-2mm]{0mm}{1mm}\\
%
\multirow{2}{*}{14} &\multirow{2}{*}{$[5^{3}, 1^{11}]$}&\multirow{2}{*}{1} & \multirow{2}*{\minitab[c]{$n_1$ \\ 11}}  & \multirow{2}*{\minitab[c]{$n_5$\\ 3}}  &  &  &
                                 \multirow{2}*{\minitab[c]{$m_{1,5}$ \\11}} &\multirow{2}*{\minitab[c]{$m_{5,5}$ \\ 2}} && &\multirow{2}{*}{326}  \rule[-2mm]{0mm}{1mm}\\
&& &&&&&&&&&   \rule[-2mm]{0mm}{1mm}\\
%
\multirow{2}{*}{15} &\multirow{2}{*}{$[5^{3}, 2^{1}, 1^{11}]$}&\multirow{2}{*}{3} & \multirow{2}*{\minitab[c]{$n_1$ \\ 11}}  & \multirow{2}*{\minitab[c]{$n_2$\\ 1}}  &\multirow{2}*{\minitab[c]{$n_5$\\ 3}} &  &\multirow{2}*{\minitab[c]{$m_{1,5}$ \\ 11}} &\multirow{2}*{\minitab[c]{$m_{2,5}$ \\ 2}}&\multirow{2}*{\minitab[c]{$m_{5,5}$ \\ 1}}  & &\multirow{2}{*}{394}  \rule[-2mm]{0mm}{1mm}\\
&& &&&&&&&&&   \rule[-2mm]{0mm}{1mm}\\
\multirow{2}{*}{16} & \multirow{2}{*}{$[5^{3}, 3^{1}, 1^{12}]$} &\multirow{2}{*}{4}&\multirow{2}*{\minitab[c]{$n_1$ \\ 12}} & \multirow{2}*{\minitab[c]{$n_3$\\ 1}} & \multirow{2}*{\minitab[c]{$n_5$\\ 3}} &&                             \multirow{2}*{\minitab[c]{$m_{1,3}$ \\2}} &\multirow{2}*{\minitab[c]{$m_{1,5}$ \\10}} & \multirow{2}*{\minitab[c]{$m_{3,5}$ \\ 1}} &\multirow{2}*{\minitab[c]{$m_{5,5}$ \\ 2}} &\multirow{2}{*}{414}  \rule[-2mm]{0mm}{1mm}\\
&& &&&&&&&&&  \rule[-2mm]{0mm}{1mm}\\
\multirow{2}{*}{17} & \multirow{2}{*}{$[5^{3}, 4^{1}, 1^{13}]$}&\multirow{2}{*}{4}&\multirow{2}*{\minitab[c]{$n_1$ \\ 13}} & \multirow{2}*{\minitab[c]{$n_4$\\ 1}} &\multirow{2}*{\minitab[c]{$n_5$\\ 3}}  &&                             \multirow{2}*{\minitab[c]{$m_{1,4}$ \\1}} &\multirow{2}*{\minitab[c]{$m_{1,5}$ \\ 12}} &\multirow{2}*{\minitab[c]{$m_{4,5}$ \\3}}  &&\multirow{2}{*}{452}  \rule[-2mm]{0mm}{1mm}\\
&& &&&&&&&&&  \rule[-2mm]{0mm}{1mm}\\
\multirow{2}{*}{18} &\multirow{2}{*}{$[5^{4}, 1^{14}]$} &\multirow{2}{*}{2} & \multirow{2}*{\minitab[c]{$n_1$ \\ 14}}  & \multirow{2}*{\minitab[c]{$n_5$\\ 4}}  & &
                                &\multirow{2}*{\minitab[c]{$m_{1,5}$ \\ 14}} &\multirow{2}*{\minitab[c]{$m_{5,5}$ \\ 3}}&& &\multirow{2}{*}{514}  \rule[-2mm]{0mm}{1mm}\\
&& &&&&&&&&&   \rule[-2mm]{0mm}{1mm}\\
\multirow{2}{*}{19} &\multirow{2}{*}{$[5^{4}, 2^{1}, 1^{14}]$}&\multirow{2}{*}{7} & \multirow{2}*{\minitab[c]{$n_1$ \\ 14}}  & \multirow{2}*{\minitab[c]{$n_2$\\ 1}}  &\multirow{2}*{\minitab[c]{$n_5$\\ 4}}  &  &
                                 \multirow{2}*{\minitab[c]{$m_{1,2}$ \\1}} &\multirow{2}*{\minitab[c]{$m_{1,5}$ \\ 13}} &\multirow{2}*{\minitab[c]{$m_{2,5}$ \\ 1}}& \multirow{2}*{\minitab[c]{$m_{5,5}$ \\ 3}} &\multirow{2}{*}{522}  \rule[-2mm]{0mm}{1mm}\\
&&&&&&&&&&&   \rule[-2mm]{0mm}{1mm}\\
\multirow{2}{*}{20} & \multirow{2}{*}{$[5^{4}, 3^{1}, 1^{15}]$}&\multirow{2}{*}{8}&\multirow{2}*{\minitab[c]{$n_1$ \\ 15}} & \multirow{2}*{\minitab[c]{$n_3$\\ 1}} & \multirow{2}*{\minitab[c]{$n_5$\\ 4}}  &&                             \multirow{2}*{\minitab[c]{$m_{1,3}$ \\2}} &\multirow{2}*{\minitab[c]{$m_{1,5}$ \\ 13}} & \multirow{2}*{\minitab[c]{$m_{3,5}$ \\ 1}} &\multirow{2}*{\minitab[c]{$m_{5,5}$ \\ 3}} &\multirow{2}{*}{542}  \rule[-2mm]{0mm}{1mm}\\
&& &&&&&&&&&  \rule[-2mm]{0mm}{1mm}\\
\hline
\caption[Extremal trees]{
Extremal trees of order up to  $20$ and maximal degree $\Delta=5$.
First column: the order of a tree.
Second column: the degree sequences that maximize the $F$-index and the number of corresponding trees, obtained by Sage.
Third column:  solutions of an integer linear programming problem with one realization of each degree sequence,  obtained by Matlab.
Fourt column: the value of the $F$-index.
}
\label{tableTreesD5}
\end{longtable}
\end{center}
%
\begin{figure}[ht!b]
\begin{center}
\includegraphics[scale=1.0]{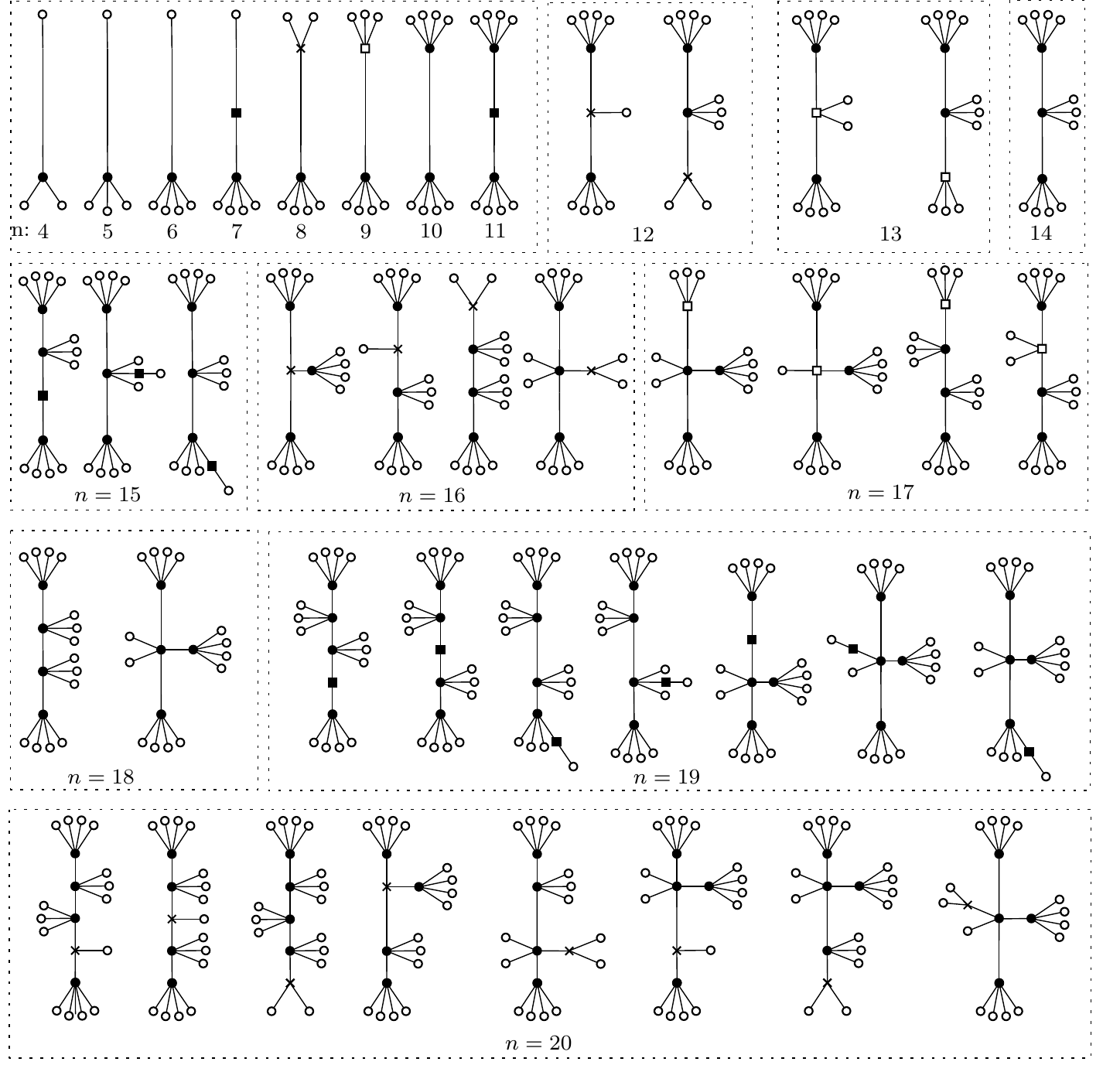}
\caption{Extremal trees of order $n = 4,5, \cdots, 20$ and  $\Delta = 5$ with respect to the $F$-index.}
\label{TreesDelta5}
\vskip-0.5cm
\end{center}
\end{figure}

\end{document}